\documentclass{article}
\usepackage[utf8]{inputenc}

\usepackage{amscd,latexsym,amssymb,amsmath,bigints,mathtools,listings}
\usepackage{enumitem}
\usepackage[toc]{appendix}
\usepackage[a4paper, total={6in, 8in}]{geometry}

\newtheorem{Thm}{Theorem}[section]
\newtheorem{Lem}{Lemma}[section]
\newtheorem{Lem*}{Lemma}

\newtheorem{Prop}[Thm]{Proposition}
\newtheorem{Prop*}{Proposition}

\newcommand{\R}{{{\mathbb R}}}

\newcommand{\E}{{{\mathbb E}}}

\newcommand{\N}{{{\mathbb N}}}

\usepackage{color}

\title{A lower bound for the spectral gap of the conjugate Kac process with 3 interacting particles}
\author{Luís Simão Ferreira\thanks{Supported by CMAF-CIO grant UIDP/04561/2020.}}
\date{\today}

\begin{document}
\maketitle

\begin{abstract}
    In this paper, we proceed as suggested in the final section of \cite{spectralgap} and prove a lower bound for the spectral gap of the conjugate Kac process with 3 interacting particles. This bound turns out to be around $0.02$, which is already physically meaningful, and we perform Monte Carlo simulations to provide a better empirical estimate for this value via entropy production inequalities. This finishes a complete quantitative estimate of the spectral gap of the Kac process.
\end{abstract}

\tableofcontents

\section{Introduction}

$N$-particle systems have been a target of flourishing research for the better part of the last 50 years, in particular due to their connection to classical statistical physics PDE's via hydrodynamic limits. Thoroughly inspecting these collision processes allows us to better understand limiting cases such as the Boltzmann equation, while avoiding dealing with the equation explicitly, since it involves strong non-linearity. 

In 1956, Marc Kac \cite{kac} introduced what we refer to as the \textit{Kac process}, a reversible Markov process for $N$ spatially homogeneous particles with one dimensional velocities, which describes the evolution of the system as particles randomly collide in a momentum and energy conserving fashion. 

A behaviour exhibited in many naturally occurring particle systems, and a desirable property of these models, is that of exponential decay to equilibrium. The easiest way one can prove such property is by showing the existence of a spectral gap for the generator of the process, and a uniform bound lets us draw conclusions as the number of particles becomes infinite. Recall that the spectral gap of an operator is defined as the infimum of its positive spectrum, which may coincide with the lowest non-zero eigenvalue, if these exist.

In essence, under the Kac process, if $L$ is the infinitesimal generator and $f_0$ is an initial distribution in our state space, then it evolves as
$f_t = e^{tL}f_0$,
and the existence of a spectral gap $\Delta$ implies an exponential decay with rate at least equal to $\Delta$.

It is hard to relate relaxation in the $N$-particle system to relaxation in the limit equation, unless one has uniform bounds on the spectral gap, independent of $N$. Kac conjectured this to be true in his toy model, and this was successfuly proven later by Janvresse \cite{janvresse}, but with one dimensional velocities and uniform jump rates, which nevertheless allows us to efficiently generate uniform points on the $n$-sphere. 

In a series of papers, Carlen, Carvalho and Loss \cite{spectralgap} managed to prove the existence of a spectral gap for a one parameter family of Kac processes in 3-dimensions, undergoing more general collisions with rates that are not bounded away from zero. We will explore next the background of these papers.

\subsection{The Kac process}

More precisely, consider an additional physical parameter $0\leq\alpha\leq2$ and an $N$-tuple of velocities $v_j$ in $\R^3$ such that
\begin{equation}
    \frac{1}{N}\sum_{j=1}^N v_j = 0, \frac{1}{N}\sum_{j=1}^N |v_j|^2 = 1,
\end{equation}
meaning that the total momentum is zero and the average kinetic energy per particle (with mass 2) is 1. These constraints describe our state space $\mathcal{S}_N$, isometric to a sphere with $3N-4$ dimensions and radius $\sqrt{N}$.

Before we move on, it is important to fix some notation. We will refer to elements of this state space $\mathcal S_N$ as a single letter $v$, $y$, and sometimes as a vector $\vec v$, $\vec y$, if it's not clear from context. Letters with subscripts such as $v_j$ shall always denote velocities in $\R^3$.

When the process begins, there is an associated exponential clock $T_{i,j}$ for each pair of velocities $(v_i,v_j)$, with parameter
\begin{equation}
    \lambda_{i,j} = N \binom{N}{2}^{-1} |v_i - v_j|^\alpha,
\end{equation}
and the minimum 
\begin{equation}
    T := \min_{i,j} T_{i,j}
\end{equation}
yields our first collision time (and colliding pair). Given the restrictions above, these collisions can be parametrized in terms of a vector $\sigma \in \mathbb{S}^2$:
\begin{align}
    v_i^*(\sigma) & = \frac{v_i+v_j}{2} + \frac{|v_i-v_j|}{2}\sigma, \\
    v_j^*(\sigma) & = \frac{v_i+v_j}{2} - \frac{|v_i-v_j|}{2}\sigma.
\end{align}
Another source of randomness on the system is then the determination of this direction vector $\sigma$. The particular choice we follow is uniform sampling, as it turns out to be more physically relevant, even though other distributions or rules could be used.

The cases $0\leq \alpha \leq 1$, with 0 being the Maxwellian molecule case and 1 the hard sphere case, are the ones of physical interest, but it turns out, as shown by Villani (regarding entropy production inequalities) \cite{sometimes} and as seen in the work of Carlen, Carvalho and Loss, that the "super hard sphere" $\alpha=2$ case provides powerful insights for lower values of the parameter, and will later be our main focus.

Their proof is based on chaos estimates that quantify the "near independence" of a finite set of coordinates. If $\sigma_N$ denotes the uniform measure on the sphere $\mathbb{S}^N(\sqrt{N})$ with radius $\sqrt{N}$, it is known that
\begin{equation}
    \lim_{N\to+\infty} \int_{\mathbb{S}^N(\sqrt{N})} \phi(v_1,...,v_k) d\sigma_N
    = \int_{\R^k}  \phi(v_1,...,v_k) d\gamma^{\otimes k},
\end{equation}
with $\gamma$ being the Gaussian measure in $\R$, and we say that the measure $\sigma_N$ is $\gamma$-chaotic. In general, if the above equality holds for $\mu$ and a sequence of measures $\mu_N$, $\mu_N$ is said to be $\mu$-chaotic.

Another important step of the proof is the introduction of a conjugate process, which allows for a clever induction argument. With the right choice of coordinates, integrals over the state space can be factored as
\begin{equation}\label{measurefactor}
    \int_{\mathcal{S}_N} \phi(\vec v) d\sigma_N = \int_B \bigg[ \int_{S_{N-1}} \phi(T_j(\vec y, v)) d\sigma_{N-1} \bigg] d\nu_N(v),
\end{equation}
with $B$ the unit ball in $\R^3$. $T_1: \mathcal{S}_{N-1}\times B \to \mathcal{S}_N$ is given by
\begin{equation}
T_1(y,v) = \bigg( \sqrt{N-1}v, \beta(v)y_1 - \frac{1}{\sqrt{N-1}}v, ..., \beta(v)y_{N-1} - \frac{1}{\sqrt{N-1}}v \bigg),
\end{equation}
while $T_j$ is defined analogously but with $v$ going to the $j$-th coordinate, and
\begin{equation}
    \beta^2(v) := \frac{N}{N-1}(1-|v|^2).
\end{equation}
Finally, the measure $\nu_N$ on $B$ is defined as
\begin{equation} \label{equilibriumradial}
    d\nu_N(v) = \frac{|\mathbb S^{3N-7}|}{|\mathbb S^{3N-4}|} (1-|v|^2)^{(3N-8)/2}dv,
\end{equation}
which is what we refer to as the equilibrium radial measure, as it is the radial marginal of the velocities for the uniform measure in $\mathcal S_N$.

As we will verify later, this formula allows for an inductive inequality of the type
\begin{equation}
    \Delta_{N,\alpha} \geq \frac{N}{N-1}\Delta_{N-1,\alpha} \hat \Delta_{N,\alpha},
\end{equation}
where $\hat \Delta_{N,\alpha}$ is the spectral gap of a \textit{conjugate Kac process}, with vastly simplified dynamics. Furthermore, the authors make use of this fact to prove that there are $K_0 > 0$ and $N_0 \in \N$ such that
\begin{equation}
    \Delta_{N,\alpha} \geq K_0 \Delta_{N_0, \alpha} > 0,
\end{equation}
leaving only the base case to show. For $N=2$, the process is trivial, as the constraints completely determine one velocity from the other. Finally, making use of a compactness argument to show that $\hat \Delta_{3,\alpha}>0$, it follows that
\begin{equation}
    \Delta_{3,\alpha} \geq \frac32 \Delta_{2,\alpha} \hat \Delta_{3,\alpha} > 0.
\end{equation}

However, unlike previous bounds that could be explicitly computed, this does not give any quantitative information about $\hat \Delta_{3,\alpha}$, besides positivity, and it remains to be shown that this gap is physically meaningful when interpreted as an expected relaxation time. In this paper, we compute a lower bound for the spectral gap of the conjugate Kac process with 3 particles.

Even though the proof provided in \cite{spectralgap} makes use of asymptotic methods in the form of quantitative chaos estimates, the results are strong enough to hold already for all $N\geq 4$, and provide a positive gap. This fails for $N=3$.

Due to the fact that $N=2$ is already trivial, and most methods allow for good estimates for $N>3$, one would expect this case to be accessible. Unfortunately, this is not true, as most methods fail to provide a positive gap in this number of dimensions.

\subsection{Structure of the paper}
This paper is organized as follows.

In section \ref{conjkac}, we introduce our main object of study, the conjugate Kac process, building up the necessary tools and finally stating our main result. We also explore the possibility of obtaining spectral gap bounds via entropy production inequalities and, while not being able to prove it analytically, present some numerical evidence of such results.

The proof is developed over section \ref{mainproof}. We begin by inspecting the spectrum of our main tool, the correlation operator $K$, and proving suitable upper bounds on the eigenvalue sequence $\kappa_{n,\ell}$. 

The proof itself is then broken down over the anti-symmetric and symmetric sectors, as they are orthogonal to each other. The first one is handled fairly quickly, while the symmetric sector, which is of bigger physical interest due to its relation to the limit equation, requires a more hands-on approach. We treat different angular momentum sectors, where different methods handle 'large' and 'small' values of $\ell$, leaving only a finite amount of eigenvalues to be computed.

\clearpage

\section{The conjugate Kac process}\label{conjkac}

We will now define the family of reversible Markov jump processes on $\mathcal S_N$, conjugate to the Kac process. Fixing a number of interacting particles $N$, $0<\alpha\leq2$, and given $\vec v \in S_N$, let $\{\hat T_1, ..., \hat T_N\}$ be a set of $N$ independent exponential variables with parameter $\lambda_k (\vec v)$ of $\hat T_k$ given by
\begin{equation}
    \lambda_k (\vec v) = \frac1N \bigg[ \frac{N^2-(1+|v_k|^2)N}{(N-1)^2} \bigg]^{\alpha/2}.
\end{equation}
Due to the restraints on energy and momentum, the maximum of $|v_k|^2$ is $N-1$, and thus the parameter is always positive.

Let $\hat T = \min \{\hat T_1, ..., \hat T_N\}$ be the first jump time, and let $k$ be the corresponding index. At this time, the process makes a conditional jump to uniform to a new point $\vec v \in S_N$, fixing $v_k$, and all other velocities are re-sampled uniformly, conditional on $v_k$. Even though this process is trivial for $N=2$, since the momentum equation completely determines one velocity from the other, our $N=3$ case is subtle. 

The generator of this process can be obtained via projection (or conditional expectation) operators. Recall that, for any $\phi \in L^2(\sigma_N)$, there is a unique element $f(v_k)$ of $L^2(\sigma_N)$ such that
\begin{equation}
    \int_{S_N}\phi(\vec v) g(v_k) d\sigma_N = \int_{S_N} f(v_k) g(v_k) d\sigma_N,
\end{equation}
and we define this way the conditional expectation of $\phi$, given $v_k$:
\begin{equation}
    P_k \phi := \E [\phi | v_k].
\end{equation}
This can also be understood as an orthogonal projection onto the subspace of square integrable functions that depend only on the $k$-th coordinate. Furthermore, the factorization formula \eqref{measurefactor} gives us an expression for this operator:
\begin{equation}
    P_k \phi (v) = \int_{\mathcal S_{N-1}} \phi (T_k(y, v_k/\sqrt{N-1}) d\sigma_{N-1}.
\end{equation}

The generator of this process can then be written as
\begin{equation}
    \hat L_{N,\alpha} f = -\frac1N \sum_{k=1}^N \frac1N \bigg[ \frac{N^2-(1+|v_k|^2)N}{(N-1)^2} \bigg]^{\alpha/2} [f-P_k f],
\end{equation}
and it is also useful to define
\begin{equation}
    w_{N}(v_k) := \frac{N^2-(1+|v_k|^2)N}{(N-1)^2}.
\end{equation}

It is easy to prove self-adjointness, as we would expect from the reversibility of the process, and its null space is spanned by the constants. Similarly as before, the spectral gap of the conjugate Kac process is defined as
\begin{equation}\label{conjgap}
    \hat \Delta_{N,\alpha} := \inf \{ \mathcal{D}_{N,\alpha} (f,f) | \langle f,1 \rangle_{L^2(\sigma_N)}, \|f\|^2_{L^2(\sigma_N)} = 1 \},
\end{equation}
with
\begin{align}\label{conjugatedirichlet}
    \mathcal{D}_{N,\alpha} (f,f) & := - \langle f, \hat L_{N,\alpha} f \rangle_{L^2(\sigma_N)}\nonumber \\
    & = \frac1N \sum_{k=1}^N \frac1N \int_{S_N} \bigg[ \frac{N^2-(1+|v_k|^2)N}{(N-1)^2} \bigg]^{\alpha/2} [f^2-f P_k f] d\sigma_N.
\end{align}

We are at this moment in conditions of explaining how this process arises. Let $L_{N,\alpha}$ be the infinitesimal generator of the Kac process,
\begin{equation}
    \mathcal{E}_{N,\alpha}(f,f) = -\langle f, L_{N,\alpha} f \rangle
\end{equation}
the Dirichlet form associated with $L_{N,\alpha}$ and $\mathcal{E}_{N,\alpha}(f,f|v_k)$ the \textit{conditional Dirichlet form} obtained by factorizing, so that
\begin{equation}
     \mathcal{E}_{N,\alpha}(f,f) = \frac{N}{N-1} \bigg( \frac1N \sum_{k=1}^N \int_B w_N^{\alpha/2}(v_k) \mathcal{E}_{N,\alpha}(f,f|v_k) d\nu_N(v_k/\sqrt{N-1}) \bigg).
\end{equation}
Then, using the spectral gap for $N-1$ particles, one has
\begin{equation}
     \mathcal{E}_{N,\alpha}(f,f) \geq \frac{N}{N-1} \Delta_{N-1,\alpha} \bigg( \frac1N \sum_{k=1}^N \int_{\mathcal{S}_N} w_N^{\alpha/2} (v_k) [f-P_k f]^2 d\sigma_N\bigg),
\end{equation}
and the quantity in brackets corresponds exactly to \eqref{conjugatedirichlet}.

Now, define the self-adjoint operators
\begin{equation}
    W^{(\alpha)} = W^{(\alpha)}(v) := \frac1N \sum_{k=1}^N w_N^{\alpha/2}(v_k) 
\end{equation}
and
\begin{equation}
    P^{(\alpha)} = \frac1N \sum_{k=1}^N w_N^{\alpha/2}(v_k) P_k,
\end{equation}
an average of weighted projections, so that we have $\hat L_{N,\alpha} = W^{(\alpha)} - P^{(\alpha)}$. For $\alpha=2$ and $N=3$, we see that $W^{(2)}=\frac34$, and \eqref{conjgap} becomes
\begin{equation}
    \hat \Delta_{3,2} = \frac34 - \mu_3,
\end{equation}
with
\begin{equation}
    \mu_3 := \sup \{ \langle f, P^{(2)} f \rangle_{ L^2(d\sigma_3) } | \|f\|_{ L^2(d\sigma_3) } = 1, \langle f, 1 \rangle_{ L^2(d\sigma_3) } = 0 \}.
\end{equation}
The spectrum of $P^{(\alpha)}$ is also studied on \cite{kspectrum}. More importantly, for $\alpha=2$, the following is true

\begin{Lem}
    The essential spectrum of $P^{(2)}$ is the interval $[0,\frac12]$, and the remaining spectrum can only consist of isolated eigenvalues in $\big(\frac12,\frac34\big)$.
\end{Lem}
Note that, in particular, $\frac34$ cannot be an accumulation point, and thus this proves that the spectral gap exists. Moreover, if $\mu_3\leq\frac12$, we immediately have that $\hat \Delta_{3,2}\geq0.25$, and if $\frac12 < \mu_3 < \frac34$, it is an eigenvalue of $P^{(2)}$. 

We then turn to the eigenvalue equation $P^{(2)} f = \lambda f $ to obtain the only missing estimate for a fully quantitative result on the three particle spectral gap $\hat \Delta_{3,2}$. In \cite{spectralgap}, the authors break this down to the \textit{symmetric} and \textit{anti-symmetric} eigen-sectors, as they are mutually orthogonal. A computation for the anti-symmetric case is provided, and a suggestion is made on how to proceed for the remaining case. In section \ref{antisymsec} we explain and slightly improve the former, while the latter is treated in section \ref{symsec} - this last case proved to be much harder, due to the analytical difficulties in dealing with the spectrum of the conditional expectation operator (see section \ref{corroperator}). 

Using this method, we were able to prove the following main result:
\begin{Thm}\label{maintheorem}
The spectral gap of the conjugate Kac process in 3 dimensions, with $\alpha = 2$, is at least $0.01984$.
\end{Thm}

The actual computations we provide show a result slightly less then $0.02$ in order to simplify and reduce the amount of cases, but could easily be adapted to achieve this quantity at the expense of explicitly determining a larger number of eigenvalues. Note that it is expected that the true value of the gap is larger, and we will illustrate this next.

\subsection{Spectral gap bounds via entropy production inequalities}

Cedric Villani \cite{sometimes} was the first to prove entropy production inequalities for variants of the Kac process, namely with one dimensional velocities. 
Even though obtaining such bounds for the 3-dimensional Kac process is still an open problem, one can adopt a similar reasoning and attempt to derive entropy production inequalities for the simpler conjugate process. Recall that the entropy functional of a suitable density $f$ on $\mathcal S_N$ is defined as
\begin{equation}
    H(f) := \int_{\mathcal{S}_N} f\log f d\sigma_N,
\end{equation}
and the entropy production of a process with generator $\hat L_{N,\alpha}$ as 
\begin{equation}
    D_{\hat L_{N,\alpha}}(f) := -\frac{d}{dt}\bigg ( H( e^{t\hat L_{N,\alpha}}f ) \bigg ) \geq 0,
\end{equation}
so that we can look for estimates of the type
\begin{equation}
    D_{\hat L_{N,\alpha}}(f) \geq C_{N,\alpha} H(f),
\end{equation}
with $C_{N,\alpha}$ a constant independent of $f$, and depending on $\hat L_{N,\alpha}$ only via the parameters describing the process.

Results of the above type are stronger than spectral gap estimates, as a classical linearization procedure implies the latter. More precisely, taking $f := 1 + \varepsilon h$, with $\varepsilon>0$ and $$\int_{\mathcal{S}_N} h d\sigma_N= 0,$$and expanding in first order leads us to
\begin{equation}
    \hat \Delta_{N,\alpha} \geq \frac{1}{2} C_{N,\alpha}.
\end{equation}
One advantage of the entropy production approach is that we can actually measure the empirical entropy in Monte Carlo simulations, and obtain an estimate of the true spectral gap by analyzing the log-entropy decay, which we will do next.

First, by symmetrizing one can compute the entropy production functional of the conjugate process:
\begin{equation}
    D_{N,\alpha}(f) := \frac1N \sum_{k=1}^N \int_{\mathcal S_N} w_N^{\alpha/2}(v_k) \log f [f - P_k f] d\sigma_N.
\end{equation}

Specializing to $\alpha=2$, and using the convexity of $-\log t$ and the fact that 
\begin{equation}
    \frac{1}{N}\sum_{k=1}^N w_N(v) = 1,
\end{equation}
we obtain
\begin{align}
    D_{N,2} (f) \geq \bar D_{N,2} (f) := & 
    \frac1N \sum_{k=1}^N \int_{\mathcal S_N} w_N(v_k) [f \log f  - P_k f \log P_k f] d\sigma_N \\
    = & H(f) - \frac{1}{N} \sum_{k=1}^N \int_{\mathcal S_N} w_N(v_k) [P_k f \log P_k f - P_k f + 1] d\sigma_N.
\end{align}

Now, as $t \log t - t + 1 \geq 0$ for all $t > 0$ and $w_N(v) \leq \frac{N}{N-1}$ for all $v$, we see that the $k$-th summand in the second term is less than $\frac{N}{N-1} H(P_k f)$, and thus we have

\begin{equation}
    \bar D_{N,2} (f) \geq H(f) - \frac{1}{N-1} \sum_{k=1}^N H(P_k f)
\end{equation}

To continue, we need to somehow relate the entropy of the marginals $P_k f$ with the entropy of the joint distribution $f$. Fortunately, this is  precisely developed in \cite{sphericalmarginals} and \cite{youngsn}, and we make use of the CCL inequality (in 3 dimensions),
\begin{equation}\label{CCL}
    \frac{N-1}{N} \sum_{k=1}^N H(P_k f) \leq 2 H(f),
\end{equation}
thus giving us
\begin{equation}
    \bar D_{N,2} (f) \geq \bigg(1-\frac{2N}{(N-1)^2}\bigg) H(f).
\end{equation}

As for the much simpler case of $\alpha=0$ with constant jump rates, one can easily verify by slightly adapting the computations above that this yields
\begin{equation}
    C_{N,0} = \frac{N-3}{N-1}.
\end{equation}

With this quick argument, we can then state the following:
\begin{Prop}
    The entropy production functional of the Kac process in 3 dimensions, $D_{N,\alpha}$, satisfies
    \begin{equation}
        D_{N,2} (f) \geq \bigg(1-\frac{2N}{(N-1)^2}\bigg) H(f),
    \end{equation}
    and 
    \begin{equation}
        D_{N,0} (f) \geq \frac{N-3}{N-1} H(f).
    \end{equation}

    Thus, for $N \geq 4$, the spectral gap of the conjugate Kac process in 3 dimensions is at least $\frac{1}{18}$ for the super hard sphere case, and $\frac{1}{6}$ for the Maxwellian molecule case.
\end{Prop}

Sadly, this is not quite enough, and the bounds degenerate at $N=3$, and adapting Villani's approach seems to give degenerate bounds at $N=3$ as well, even though the technique is much stronger than the one we use above. The main culprit seems to be the CCL inequality \eqref{CCL} as the extra $N/(N-1)$ factor, which isn't present for the 1-dimensional case, loses us enough for the method to fail.

This makes it clear that we're dealing with an edge case, just outside the reach of most methods, even if we can simultaneously handle both the entropy production and the spectral gap problems for all $N$ above 3.

Thus, a careful analysis of the dynamics and geometry of the problem is needed to obtain such a bound, and we make use of the information we have about the spectrum of the conditional expectation operator $K$ on the sphere and properties of orthogonal polynomials for this purpose.

\subsection{Monte Carlo simulations}

In order to observe this relaxation to equilibrium, we perform Monte Carlo simulations of the Kac process with 3 interacting particles. Even if we can't prove entropy production inequalities for the conjugate Kac process when $N=3$, this method allows us to obtain an estimate of 0.3 for the exponential entropy decay rate, and thus an estimate of about 0.15 for the spectral gap, which is only about an order of magnitude larger than the theoretical result we achieve.

Based on the factorization formula \eqref{measurefactor}, we notice that the uniform measure in $S_3$, $\sigma_3$ factorizes as
$$ d\sigma_3 = d\sigma_2 \times d\nu_3, $$
with
$$ d\nu_3 := \frac{|\mathbb{S}^2|}{|\mathbb{S}^5|} (1-|v|^2)^\frac12 dv,$$
induced by the transformation $T_1: S_2 \times B \to S_3$, 
$$T_1(y,v) := \bigg(\sqrt{2}v, \beta(v)y_1-\frac{v}{\sqrt{2}}, \beta(v)y_2-\frac{v}{\sqrt{2}}\bigg),$$
$$ \beta (v) := \sqrt{\frac32 (1-|v|^2)}.$$

As the system decays to equilibrium and the distribution of the velocities approaches uniform, the marginal distributions become closer do $d\nu_3$. In particular, by symmetry, we will only look at the radial component of this measure, which is proportional to $r^2(1-r^2)^\frac12 dr$. Additionally, the factorization formula and its symmetries reduce the conditional sampling problem to sampling from $S_2$, which is trivial. We refer to the radial distribution of the first particle as the sampled distribution, and the others as the implied distributions.

Thus, we can proceed as follows. Fixing the first coordinate, we sample $r = |v|$ from a given distribution using an acceptance/rejection mechanism, and sample $y_1$, $y_2$ uniformly, so that the point is in $S_3$. We then let the system evolve sampling first the jump time and the particle that stays fixed, identified with $\sqrt{2}v$, and sample $y_1$, $y_2$ uniformly on $S_2$.

We perform 1 million simulations of the process in Python, for $\alpha=0,2$, and use these simulations to build a histogram approximating the probability distribution over 100 bins and several timestamps. We then measure the relative entropy between the empirical measure and the equilibrium radial measure over time, and use a log-plot to showcase and estimate the exponential decay. This was done over a selection of initial distributions, and a "worst case scenario" is showcased here with the sampling distribution being $2(1-x)$, as most of the probability mass is concentrated close to zero.

In figures 1, 2 and 3, the rather quick relaxation is evident. By computing the relative entropy using a first order Euler method, we see that there is a log-linear relationship up to a certain point, where the distributions are close enough so that the numerical approximations and random noise prevent the entropy from decaying below $e^{-6}$. Meanwhile, figures 4, 5 and 6 show the evolution of the distributions over several timestamps.

\clearpage

\begin{figure}[!htb]\centering
\minipage{0.32\textwidth}
\centering
  \includegraphics[width=\linewidth]{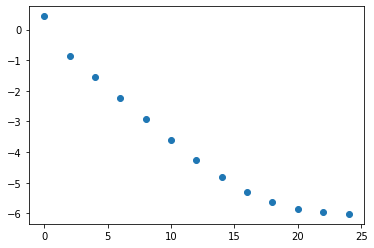}
  \caption{$\alpha=2$, sampled entropy}
  \label{fig:awesome_image1}
\endminipage\hfill
\minipage{0.32\textwidth}
\centering
  \includegraphics[width=\linewidth]{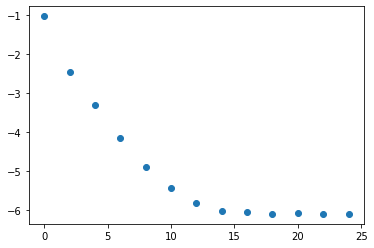}
  \caption{$\alpha=2$, implied entropy}
\endminipage\hfill
\minipage{0.32\textwidth}%
\centering
  \includegraphics[width=\linewidth]{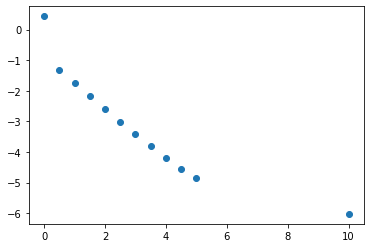}
  \caption{$\alpha=0$, sampled entropy}
\endminipage
\end{figure}

\begin{figure}[!htb]\centering
\minipage{0.32\textwidth}
\centering
  \includegraphics[width=\linewidth]{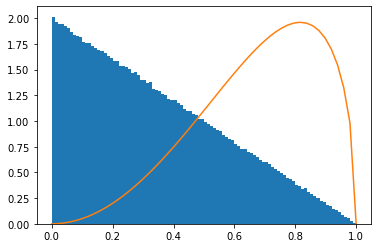}\\
  t=0s
  \label{fig:awesome_image1}
\endminipage\hfill
\minipage{0.32\textwidth}
\centering
  \includegraphics[width=\linewidth]{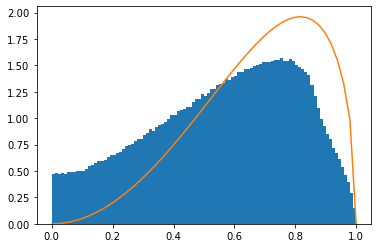}\\
  t=0.5s
  \label{fig:awesome_image2}
\endminipage\hfill
\minipage{0.32\textwidth}%
\centering
  \includegraphics[width=\linewidth]{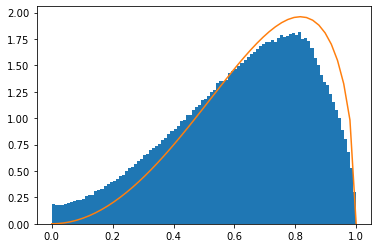}\\
  t=2s
\endminipage
\end{figure}
\begin{figure}[!htb]\centering
\minipage{0.32\textwidth}\centering
  \includegraphics[width=\linewidth]{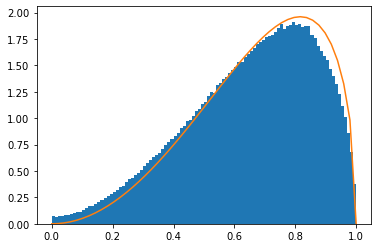}\\
  t=3.5s
\endminipage\hfill
\minipage{0.32\textwidth}
\centering
  \includegraphics[width=\linewidth]{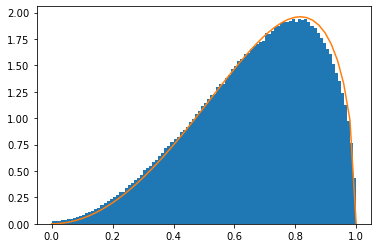}\\
  t=5s
\endminipage\hfill
\minipage{0.32\textwidth}
\centering
  \includegraphics[width=\linewidth]{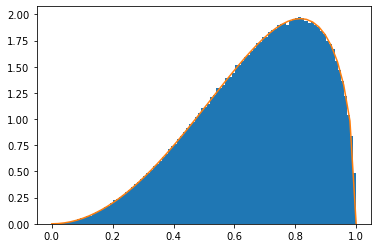}\\
  t=10s
\endminipage
\caption{$\alpha=0$, sampled velocity}
\end{figure}

\begin{figure}[!htb]\centering
\minipage{0.32\textwidth}
\centering
  \includegraphics[width=\linewidth]{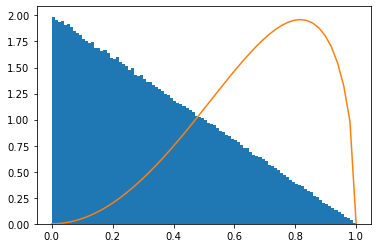}\\
  t=0s
  \label{fig:awesome_image1}
\endminipage\hfill
\minipage{0.32\textwidth}
\centering
  \includegraphics[width=\linewidth]{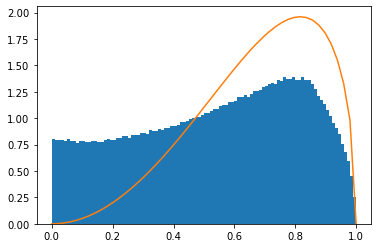}\\
  t=2s
  \label{fig:awesome_image2}
\endminipage\hfill
\minipage{0.32\textwidth}%
\centering
  \includegraphics[width=\linewidth]{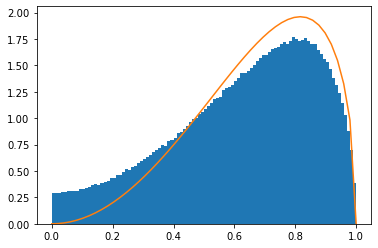}\\
  t=6s
\endminipage
\end{figure}
\begin{figure}[!htb]\centering
\minipage{0.32\textwidth}\centering
  \includegraphics[width=\linewidth]{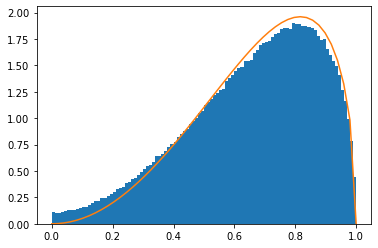}\\
  t=10s
\endminipage\hfill
\minipage{0.32\textwidth}
\centering
  \includegraphics[width=\linewidth]{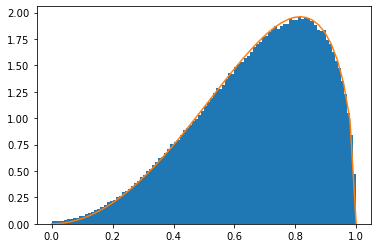}\\
  t=16s
\endminipage\hfill
\minipage{0.32\textwidth}
\centering
  \includegraphics[width=\linewidth]{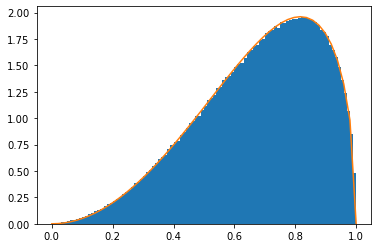}\\
  t=24s
\endminipage
\caption{$\alpha=2$, sampled velocity}
\end{figure}

\begin{figure}[!htbp]
\minipage{0.32\textwidth}\centering
  \includegraphics[width=\linewidth]{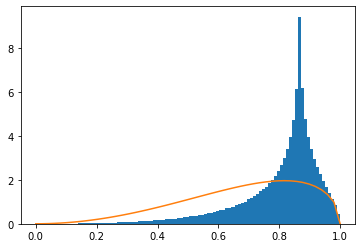}\\
  t=0s
  \label{fig:awesome_image1}
\endminipage\hfill
\minipage{0.32\textwidth}\centering
  \includegraphics[width=\linewidth]{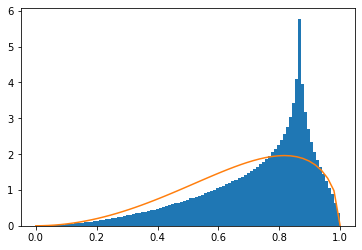}\\
  t=2s
  \label{fig:awesome_image2}
\endminipage\hfill
\minipage{0.32\textwidth}%
\centering
  \includegraphics[width=\linewidth]{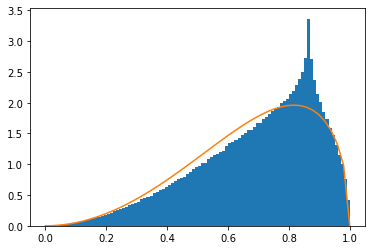}\\
  t=6s
  \label{fig:awesome_image3}
\endminipage
\end{figure}

\begin{figure}[!htbp]
\minipage{0.32\textwidth}
\centering
  \includegraphics[width=\linewidth]{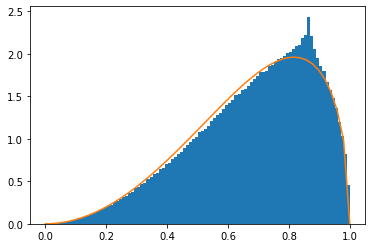}\\
  t=10s
\endminipage\hfill
\minipage{0.32\textwidth}
\centering
  \includegraphics[width=\linewidth]{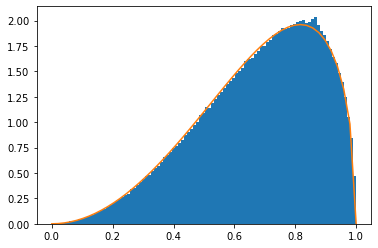}\\
  t=16s
\endminipage\hfill
\minipage{0.32\textwidth}\centering
  \includegraphics[width=\linewidth]{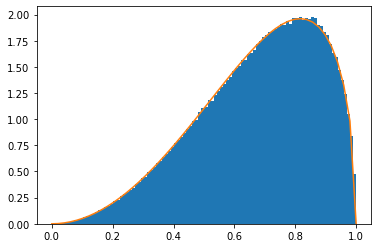}\\
  t=24s
\endminipage
\caption{$\alpha=2$, implied velocity}
\end{figure}

\clearpage

\section{Proof of the main result}\label{mainproof}

\subsection{The correlation operator on the sphere} \label{corroperator}

The main result is based on careful analysis of the eigenfunctions and the spectrum of $K$. We turn our attention to the study of this operator.

Let $B$ be the unit ball in $\R^3$, $N\geq3$, $\nu_N$ as defined in \eqref{equilibriumradial} and consider the projections $\pi_k (\vec v): S_N \to B$,
\begin{equation}
    \pi_k (\vec v) := \frac{1}{\sqrt{N-1}}v_k. 
\end{equation}
By formula \eqref{measurefactor}, we then have that
\begin{equation*}
    \int_B \psi (v) d\nu_N = \int_{S_N}\psi(\pi_k(\vec v)) d\sigma_N,
\end{equation*}
for any integrable function $\psi$ on $B$. We then define the operator $K$ on $L^2(B,\nu_N)$ by

\begin{equation}
    \langle \psi_1, K\psi_2 \rangle_{L^2(B,\nu_N)} := \int_{S_N} \psi_1^*(\pi_1(\vec v)) \psi_2(\pi_2(\vec v)) d\sigma_N.
\end{equation}
Note that the symmetries of the variables make $K$ self-adjoint. Furthermore, a simple rescaling, to extend $K$ to $L^2(\sigma_N)$, allows us to interpret this operator simply as a conditional expectation on the uniform sphere
\begin{equation*}
    K\xi (v) = \E [ \xi (v_1) | v_2 = v],
\end{equation*}
but we will not consider this scaling as the former simplifies the treatment of eigenvalues. 

Carlen, Carvalho and Loss make use of this operator to prove quantitative chaos estimates; that is, to quantify how far from independent the coordinates are, asymptotically in $N$, and most of their work is based on extensive study of the spectrum of $K$, built up over \cite{kspectrum}, \cite{markovsp}.

Before we continue, define the functions
\begin{equation}
    \eta_0 = 1, \eta_j(v) = v_j,\text{ for } 1\leq j \leq 3, \text{ and } \eta_4(v) = (|v|^2-1)/(N-1).
\end{equation}

We will now present the most important facts obtained in these papers for general $N$, and later turn to the our $N=3$ case. Following the results in \cite{kspectrum}, we have this next lemma:
\begin{Lem}
    Let $N\geq3$. The operator $K$ is compact. The function $\eta_0$ is an eigenfunction of $K$ with eigenvalue 1, and spans the corresponding eigenspace. The functions $\eta_j$, for $1\leq j \leq 4$, are eigenfunctions with eigenvalue $-\frac{1}{N-1}$, and they form an orthogonal basis for this subspace. No other eigenvalues of $K$ are larger in absolute value then $\frac{5N-3}{3(N-1)^3}$. Thus, for all $\psi \in L^2(B,\nu_N)$ orthogonal to these $\eta_j$, 
    \begin{equation}
        \|K\psi\|_2^2 \leq \frac{5N-3}{3(N-1)^3}\|\psi\|^2_2.
    \end{equation}
    Finally, every eigenvalue $\kappa$ of $K$, other than 1, $-\frac{1}{N-1}$ and $\frac{5N-3}{3(N-1)^3}$, satisfies
    \begin{equation}
        -\frac{7N-3}{3(N-1)^4} \leq \kappa \leq \frac{5N-3}{3(N-1)^3}.
    \end{equation}
\end{Lem}

Unfortunately, these bounds are not enough when $N=3$ and are just slightly off to allow us to obtain the main result more easily. On the other hand, the authors also manage to establish that $K$ has a complete basis of eigenfunctions of the form
\begin{equation}\label{basis}
    g_{n,\ell,m} = h_{n,\ell}(|v|)|v|^\ell \mathcal{Y}_{\ell,m}(v/|v|),
\end{equation}
where $\mathcal{Y}_{\ell,m}$ denotes the $\ell$-th degree spherical harmonic and, with $t=|v|^2-1$, $\alpha=\frac{3N-8}{2}\bigg ( =\frac12 \bigg)$ and $\beta = \ell + \frac12$,
\begin{equation*}
    h_{n,\ell}(|v|) = P_n^{(\alpha,\beta)}(t),
\end{equation*}
the $n$-th degree Jacobi polynomial (recall that the Jacobi polynomials are the orthogonal polynomials on the interval $[-1, 1]$ for the weight $(1-x)^\alpha (1+x)^\beta$). For convenience, we let capital $P$ denote these polynomials under Rodrigues' normalization,
\begin{equation}
    P_n^{(\alpha,\beta)}(t) = \frac{(-1)^n}{2^n n} (1-x)^{-\alpha}(1+x)^{-\beta} \frac{d^n}{dx^n}\bigg( (1-x)^{\alpha+n}(1+x)^{\beta+n}\bigg),
\end{equation}
while $p_n^{(\alpha,\beta)}$ will correspond to the $L^2$ orthonormal basis. Equation \eqref{basis} also allows us to conclude that the spectrum of $K$ consists of eigenvalues $\kappa_{n,\ell}$, independent of $m$ and given by the explicit formula
\begin{equation}
    \kappa_{n,\ell} = \frac{P_n^{(\alpha, \beta)}\Big(-1+\frac{2}{(N-1)^2}\Big)}{P_n^{(\alpha, \beta)}(1)} \bigg(-\frac{1}{N-1} \bigg)^\ell.
\end{equation}

While similar formulas derived from the above were useful in the asymptotic $N$ case, a large amount of cancellation happens and makes the eigenvalues take values much smaller than implied by naive bounds. Furthermore, it appears that the eigenvalues display an explicit $\mod 3$ behavior. By this, we mean that for each $\ell$, plotting the eigenvalues for $n \mod 3$ separately yields 3 eventually monotonous sequences, as we can observe bellow:

\begin{figure}[!htb]\centering
\minipage{0.32\textwidth}
\centering
  \includegraphics[width=\linewidth]{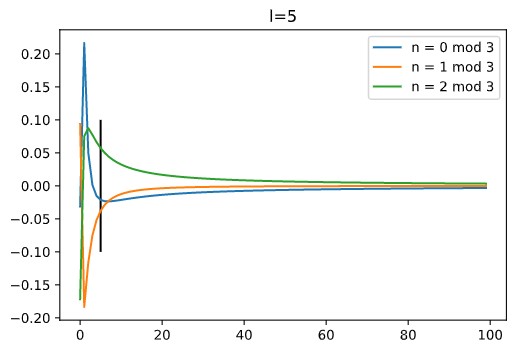}
  \caption{First 300 values of $\kappa_{n,5}$}
  \label{fig:awesome_image1}
\endminipage\hfill
\minipage{0.32\textwidth}
\centering
  \includegraphics[width=\linewidth]{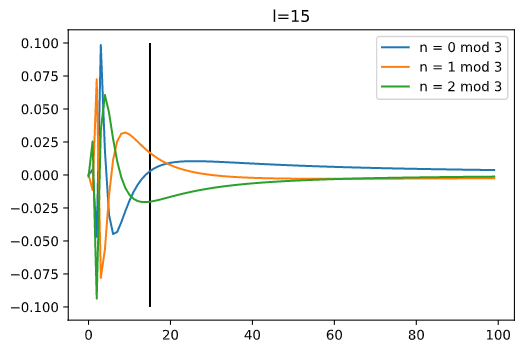}
  \caption{First 300 values of $\kappa_{n,15}$}
\endminipage\hfill
\minipage{0.32\textwidth}%
\centering
  \includegraphics[width=\linewidth]{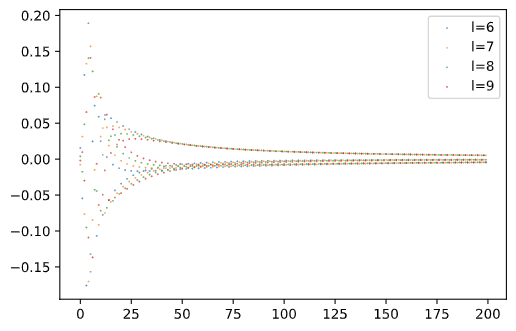}
  \caption{Overlap of the plots from $\ell=6$ to $10$}
\endminipage
\end{figure}

As these 3 different behaviours are included simultaneously in the sequences, it further makes difficult the tasks of obtaining good uniform bounds. 

One can also make use of the recurrence relations for Jacobi polynomials \cite{orthogonalpolynomials} to obtain the same type of relations for the eigenvalues $\kappa_{n,\ell}$. More precisely, for each $\ell$, we have that 
\begin{equation}\label{krecurrence}
    -\frac12 \kappa_{n,\ell} = \frac{n+\ell+\frac12}{2n+\ell+1}\kappa_{n,\ell-1} + \frac{n+\frac12}{2n+\ell+1} \kappa_{n+1,\ell-1}\ .
\end{equation}

Since we also know that
\begin{equation}\label{kzero}
    \kappa_{n,0} = \frac{1}{n+1} (-1)^n \sin\bigg(\frac{n\pi}{3}\bigg),
\end{equation}
this provides a fast and efficient way to at least numerically compute all the needed eigenvalues. We can, however, solve this recurrence relation in terms of the $\ell=0$ values, even if the obtained representation isn't very useful:

\begin{equation}
    \kappa_{n,\ell} = (-2)^\ell \sum_{j=0}^\ell \binom{\ell}{j} C_{n,\ell}^{(j)} \kappa_{n+j,0}\ ,
\end{equation}
with
\begin{equation}
    C_{n,\ell}^{(j)} = \frac{\prod_{i=1}^\ell \{n+i+1/2\}}{\prod_{i=0,i\neq j}^\ell \{2n+i+j+2\}}\ .
\end{equation}
One can see that the missing term in the denominator, when $i=j$, is exactly 2 times the numerator of $\kappa_{n+j,0}$. Unfortunately, this representation shows no indication of decay, meaning that there is a large amount of cancellation occurring (the terms are large, but very close in absolute value). Collecting the sine terms may, however, shed some light on the $\mod 3$ question.

Thus, while explicit values of Jacobi polynomials are not trivial to compute, the authors make use of uniform pointwise bounds obtained by Nevai, Erdelyi and Magnus \cite{unifjacobi}, a surprisingly deep result, to obtain monotonous bounds on these eigenvalues:

\begin{Lem}
    We have that
    \begin{equation}
        \kappa^2_{n,\ell} \leq \frac{2e}{\pi}g_1(n,\ell,N)g_2(N)g_3(n,N)g_4(n,\ell,N)\ ,
    \end{equation}
    with
    \begin{align}
        g_1(n,\ell,N) & = \bigg( \frac{4+\sqrt{9N^2-48N+65+4\ell^2+4\ell}}{3N+4n+2\ell-5} \bigg)\ ,\nonumber \\
        g_2(N) & = \bigg( \frac{(N-1)^2}{N(N-2)} \bigg)^{(3N-7)/2}\ ,\nonumber \\
        g_3(n,N) & = \frac{\Gamma(n+1)\Gamma(\frac32N-3)}{\Gamma(n+\ell+\frac32 N - \frac52)}\ ,\nonumber \\
        g_4(n,\ell,N) & =  \frac{(N-1)^2\Gamma(n+\ell+\frac32)\Gamma(\frac32N-3)}{\Gamma(n+\ell+\frac32N-\frac52)}\ .
    \end{align}
\end{Lem}

This bound can be enlarged to allow for the monotonous behaviour. This time, however, when $N=3$, we are actually in a better position to further simplify the achieved bound, yielding
\begin{equation}\label{kappahatbound}
    \kappa_{n,\ell}^2 \leq \frac{8e}{3}\frac{\Gamma(n+1)}{\Gamma(n+3/2)} \frac{\Gamma(n+\ell+3/2)}{\Gamma(n+\ell+2)} \leq \frac{8e}{3} \frac{1}{(n+1)^{1/2}} \frac{1}{(n+\ell+3/2)^{1/2}} =: \hat \kappa_{n,\ell}^2 \ .
\end{equation}

Now notice that, for all $\ell \geq 4$,
\begin{equation}
    (n+1)(n+\ell+3/2) \geq (n + \sqrt{\ell})^2\ ,
\end{equation}
so that if we define
\begin{equation}\label{kappatildebound}
\widetilde{\kappa}_{n,\ell} = \sqrt{\frac{8e}{3} }\frac{1}{\sqrt{n + \sqrt{\ell}}}\ ,
\end{equation}
we have $\hat{\kappa}_{n,\ell} \leq \widetilde{\kappa}_{n,\ell}$. Both upper bounds will be useful later on. Inequality \eqref{kappahatbound} is more precise, especially when we are looking for bounds independent of $\ell$. Nevertheless, inequality \eqref{kappatildebound} has a far simpler representation and will be considered when we need to take derivatives.

\subsection{Anti-symmetric sector}\label{antisymsec}

Now that we have built the necessary tools, we first treat the simpler anti-symmetric case, which can be done mostly through clever manipulation and simple bounds. We reproduce here the approach done in \cite{spectralgap}, with some minor numerical improvements.

\begin{Lem}
    For $N=3$, the largest eigenvalue of $P^{(2)}$ on the orthogonal complement of the symmetric sector is no larger than $0.729$. Thus, either we have $\hat{\Delta}_{3,2}\geq 0.021$ or else the gap eigenfunction is symmetric.
\end{Lem}

Suppose $f$ is anti-symmetric and we assume, without loss of generality, that $f(v)=\varphi(v_1)-\varphi(v_2)$, with $\varphi$ orthogonal to the constants. Then, $P^{(2)}f=\lambda f$ becomes
\begin{equation}
    \frac{1}{N}w_{N,2}(v_1)(1-K)\varphi(v_1) -\frac{1}{N}w_{N,2}(v_2)(1-K)\varphi(v_2)=\lambda(\varphi(v_1)-\varphi(v_2))
\end{equation}

Multiplying by $\varphi(v_1)$ and integrating over $\mathcal{S}_N$, we obtain

\begin{equation}\label{eigen2form}
    \frac{1}{N}\int_{\mathcal{S}_N} w_{N,2}(v_1)|(1-K)\varphi(v_1)|^2d\sigma_N = \lambda \langle \varphi, (1-K)\varphi\rangle.
\end{equation}
From
\begin{equation}
    w_{N,2}(v_k)=\frac{N}{N-1}-\frac{N}{(N-1)^2}|v_k|^2 ,   
\end{equation}
we can write \eqref{eigen2form} as
\begin{equation}
    \frac{1}{N-1} \langle \varphi, (1-K)^2 \varphi \rangle - \frac{1}{(N-1)^2} \langle (1-K)\varphi,|v|^2 (1-K)\varphi\rangle = \lambda \langle \varphi, (1-K)\varphi\rangle
\end{equation}

After decomposing $\sqrt{1-K}\varphi = \psi + \zeta$, where $\psi$ is orthogonal to the constants, the components $v_j$ and $|v|^2$, we have that:
\begin{align}
\langle \varphi, (1-K)^2 \varphi \rangle
& = \langle \psi + \zeta, (1-K)(\psi + \zeta) \rangle \nonumber \\
& = \langle \psi, (1-K)\psi \rangle + (1+1/(N-1))\|\zeta\|^2_2 \nonumber \\
& = \|\sqrt{1-K}\psi\|^2_2+\frac{N}{N-1}\|\zeta\|^2_2
\end{align}

and 

\begin{align}
\langle (1-K)\varphi,|v|^2(1-K)\varphi \rangle = & \langle \sqrt{1-K}(\psi+\zeta),|v|^2\sqrt{1-K}(\psi+\zeta) \rangle \nonumber \\
= & \langle \sqrt{1-K}\psi,|v|^2\sqrt{1-K}\psi \rangle + \langle \sqrt{1-K}\zeta,|v|^2\sqrt{1-K}\zeta\rangle \nonumber \\
+ &  2\langle |v| \sqrt{1-K}\psi,|v|\sqrt{1-K}\zeta \rangle,
\end{align}
since $\zeta$ is an eigenfunction of $K$ with eigenvalue $-1/(N-1)$, it is also an eigenfunction of $\sqrt{1-K}$, with eigenvalue $\sqrt{1+\frac{1}{N-1}}=\sqrt{\frac{N}{N-1}}$.

It then follows that
\begin{align}
    \langle \sqrt{1-K}\varphi,|v|^2 \sqrt{1-K}\varphi\rangle 
    & \geq \langle \sqrt{1-K}\psi,|v|^2\sqrt{1-K}\psi \rangle + \langle \sqrt{1-K}\zeta,|v|^2\sqrt{1-K}\zeta\rangle \nonumber \\
    & - 2\||v|\sqrt{1-K}\psi\||_2\sqrt{\frac{N}{N-1}}\||v|\zeta\|_2 \nonumber \\
    & \geq \bigg(1-\frac{1}{t}\bigg)\langle \sqrt{1-K}\psi,|v|^2\sqrt{1-K}\psi \rangle + (1-t)\frac{N}{N-1}\||v|\zeta\|_2^2,
\end{align}
using the Cauchy-Schwarz and arithmetic-geometric mean inequalities.

Before we move on, it is useful to note that $\||v|\zeta\|_2^2 \geq \|\zeta\|^2_2$. To see this for $N=3$, write 
\begin{equation}\label{zeta}
    \zeta = \sum_{j=1}^4 a_j \eta_j, 
\end{equation}
with $\eta_j(v) = \boldsymbol{e}_j \cdot v$ for $j=1,2,3$ and $\eta_4 = |v|^2-1$ as discussed in section \ref{corroperator}.

From \eqref{zeta}, a simple computation yields
\begin{equation}
    \| |v|\zeta\|^2_2 = \sum_{j=1}^4 |a_j|^2 \||v|\eta_j\|^2_2\ ,
\end{equation}
and while a bit more cumbersome, one can also show that
\begin{equation}
    \int_{\mathcal S_3} |v_1|^4 d\sigma_3 = \frac54
\end{equation}
and
\begin{equation}
    \int_{\mathcal S_3} |v_1|^6 d\sigma_3 = \frac74.
\end{equation}
Together, the equations above give us
\begin{equation}
    \||v|\eta_j \|^2_2 = \frac{5}{4} \text{ for } j=1,2,3 \text{ and } \||v|\eta_4\|=1,
\end{equation}
which prove what we needed.

Finally, from the trivial bound $\langle \sqrt{1-K}\psi,|v|^2\sqrt{1-K}\psi \rangle \leq (N-1)\|\sqrt{1-K}\psi\|^2_2$, we have that, for all $0<t<1$,
\begin{equation}
    \langle \sqrt{1-K}\varphi,|v|^2 \sqrt{1-K}\varphi\rangle
    \geq \bigg(1-\frac{1}{t}\bigg)(N-1)\|\sqrt{1-K}\psi\|^2_2 + (1-t)\frac{N}{N-1}\|\zeta\|_2^2\ .
\end{equation}

Specializing to $N=3$, the second most negative eigenvalue of $K$ is $-\frac{3}{8}$, which yields $\|\sqrt{1-K}\psi\|^2_2\leq \frac{11}{8}\|\psi\|^2_2$. Combining all of the above,

\begin{align}
    \lambda (\|\psi\|^2_2+\|\zeta\|^2_2)  & \leq \frac{1}{2}\cdot\frac{11}{8} \|\psi\|^2_2+\frac{3}{4}\|\zeta\|^2_2-\frac{1}{2}\cdot\frac{11}{8}\bigg(1-\frac{1}{t}\bigg)\|\psi\|^2_2-\frac{1}{4}\frac{3}{2}(1-t)\|\zeta\|^2_2 \nonumber \\
    & = \frac{11}{16t}\|\psi\|_2^2+	\frac{3}{8}\bigg(1+t\bigg)\|\zeta\|^2_2 \ .
\end{align}
With $t=0.943$, this yields $\lambda\leq 0.729$.

\subsection{Symmetric sector} \label{symsec}

The remaining task is to bound the second largest eigenvalue of $P^{(2)}$
in the symmetric sector. Assume $f$ is of the form
\begin{equation}
    f(\Vec{v})=\sum_{j=1}^N \varphi(v_j).
\end{equation}

Then, we can write $P^{(2)}f = \lambda f$ as
\begin{equation}
    \frac{1}{N}\sum_{j=1}^N w_{N,2}(v_j)(\varphi(v_j) + (N-1)K\varphi(v_j))=\lambda\sum_{j=1}^N \varphi(v_j).
\end{equation}
Theorem 2.5 in \cite{spectralgap} is here extremely important, as it allows us to treat the above sum component wise:

\begin{equation}
    \frac{1}{N}w_{N,2}(v)(\varphi(v) + (N-1)K\varphi(v))=\lambda\varphi(v).
\end{equation}
By integrating against $\varphi(v)$, it follows that
\begin{equation}\label{inteq}
    \frac{1}{N}\int_{\mathcal{S}_N}\varphi(v_1)w_{N,2}(v_1)(\varphi(v_1)+(N-1)K\varphi(v_1))d\sigma_N = \lambda \|\varphi\|^2_2.
\end{equation}

As, once again,
\begin{equation}
    w_{N,2}(v)=\frac{N}{N-1}-\frac{N}{(N-1)^2}|v|^2,
\end{equation}
and focusing on $N=3$, \eqref{inteq} becomes
\begin{eqnarray}\label{inteq2}
 \lambda \|\varphi\|^2_2   &=& \frac{1}{2}\int_{\mathcal{S}_3}\varphi(v_1)(1- \tfrac12 |v|^2)(1+2K)\varphi(v_1)d\sigma_N\nonumber\\
& = & \frac12 \langle (1- \tfrac12 |v|^2)\varphi, (1+2K) \varphi \rangle .
\end{eqnarray}

Since
\begin{equation}
    -1 \leq \bigg(1 -\frac12 |v|^2\bigg) \leq 1,\text{ }\bigg\| \bigg(1- \frac12 |v|^2\bigg)\varphi\bigg\|_2 \leq \|\varphi\|_2,
\end{equation}
we obtain the first bound
\begin{equation}\label{crude}
    \lambda \leq \frac12 \|1+2K\|_2.
\end{equation}
Unfortunately, this bound is useless as is, since the largest eigenvalue of $K$, other than $1$, is $\frac12$, giving us $\lambda \leq 1$. Thus, we need to be more incisive on the spectrum of $K$ and its interaction with the multiplication operator $(1-\frac12 |v|^2)$.

Consider the operator
\begin{equation}
    M:=|v|^2(1+(N-1)K).
\end{equation}

Recalling the eigenbasis of $K$, 
\begin{equation}
    g_{n,l,m}=h_{n,l}(|v|^2)|v|^l \mathcal{Y}_{l,m}(v/|v|),
\end{equation}
we see that $M$ commutes with rotations and thus different angular momentum sectors will be mutually orthogonal, so we can consider separately the dependence in $\ell$ and $m$ (in fact, we will ignore $m$ completely, as $g_{n,\ell,m}$, $g_{n,\ell,m'}$ belong to the same $\ell$-eigenspace). Thus, the quantities we are interested in estimating are
\begin{equation}
    \lambda_\ell := \sup_{\|\varphi\|=1} \frac14 \langle (1-t)\varphi, (1+2K) \varphi \rangle ,
\end{equation}
with the supremum being taken over the $\ell$ sector.
Moreover, as $h_{n,l}(|v|^2)=P_n^{(\alpha,\beta)}(t)$, with $t=|v|^2-1$, $\beta=l+\frac{1}{2}$, we see that $1-\frac12|v|^2 = \frac12 (1-t)$ in the Jacobi variables, and we can exploit the properties of Jacobi polynomials.

The slow simultaneous decay in $n$ and $\ell$ we obtain from expressions like \eqref{kappatildebound} pose an obstacle for computing meaningful bounds. Thus, we break down the problem over several ranges. 

First, we look for a uniform bound given large enough $\ell>\ell^*$, but such that $\ell^*$ is as small as possible. 

For $5<\ell<\ell^*$, the eigenvalues are already small enough for \eqref{crude} to work as well, and we make use of a uniform bound in $n$ to reduce this problem to explicitly computing a finite amount of these. Note that the recurrence formula (recurrence) makes this fairly trivial and efficient. 

Finally, we provide an argument based on linear algebra for $\ell \leq5$, where closed form expressions are simple enough for this.

\subsubsection{Large $\ell$ case}

Let $f_{n,\ell,m}$ be a re-scaling of $g_{n,\ell,m}$ such that the basis is orthonormal, and fix $\ell$, $m$.

Writing $P_n := P_n^{(\alpha,\beta)}$, the three-term relation for Jacobi polynomials \cite{orthogonalpolynomials} is 
\begin{equation}
    P_{n+1}(t)=(A_n t + B_n)P_n(t) + C_n P_{n-1}(t)\ ,
\end{equation}
with
\begin{align}
    A_n & = \frac{(2n+\alpha + \beta +1)(2n+\alpha+\beta+2)}{2(n+1)(n+\alpha+\beta+1)}\ , \nonumber \\
    B_n & = \frac{(\alpha^2-\beta^2)(2n+\alpha+\beta+1)}{2(n+1)(n+\alpha+\beta+1)(2n+\alpha+\beta)}\ , \nonumber \\
    C_n & = \frac{(n+\alpha)(2n+\alpha+\beta+2)(n+\beta)}{(n+1)(n+\alpha+\beta+1)(2n+\alpha+\beta)}\ ,
\end{align}
which gives
\begin{equation}
  tP_n = (1/A_n) P_{n+1} - (B_n/A_n) P_n + (C_n/A_n) P_{n-1}\ .  
\end{equation}

As we are working with orthonormal functions, we have to re-scale the equality. By the orthogonality relations for Jacobi polynomials,
\begin{equation}
    \langle P_n^{(\alpha, \beta)}, P_n^{(\alpha, \beta)}\rangle_{L^2(\mu^{(\alpha, \beta)})} = \frac{2^{\alpha+\beta+1}}{2n+\alpha+\beta+1}
    \frac{\Gamma(n+\alpha+1)\Gamma(n+\beta+1) }{\Gamma(n+\alpha+\beta+1)n!}\delta_{nm}\ ,
\end{equation}
we first compute
\begin{align}
    & \|P_n^{(\alpha,\beta)}\|_2^2 = \frac{2^{\alpha+\beta+1}}{2n+\alpha+\beta+1}
    \frac{\Gamma(n+\alpha+1)\Gamma(n+\beta+1) }{\Gamma(n+\alpha+\beta+1)n!}\ ,\nonumber \\
    & \frac{\|P_{n-1}^{(\alpha,\beta)}\|_2^2}{\|P_n^{(\alpha,\beta)}\|_2^2} = \frac{2n+\ell+2}{2n+\ell} \frac{n(n+\ell+1)}{(n+1/2)(n+\ell+1/2)}\ , \nonumber\\
    & \frac{\|P_{n+1}^{(\alpha,\beta)}\|_2^2}{\|P_{n}^{(\alpha,\beta)}\|_2^2} = \frac{2n+\ell+2}{2n+\ell+4}\frac{(n+3/2)(n+\ell+3/2)}{(n+1)(n+\ell+2)}\ . 
\end{align}
Defining the scaling factors
\begin{align}
    F_n^-  & := \frac{\|P_{n-1}\|}{\|P_n\|} = \sqrt{\frac{n(2n+\alpha+\beta+1)(n+\alpha+\beta)}{(n+\alpha)(n+\beta)(2n+\alpha+\beta-1)}}\ , \nonumber \\
    F^+_n & := \frac{\|P_{n+1}\|}{\|P_n\|} = \sqrt{\frac{(n+\alpha+1)(n+\beta+1)(2n+\alpha+\beta+1)}{(n+1)(n+\alpha+\beta+1)(2n+\alpha+\beta+3)}}\ ,
\end{align}
the three-term relation for the multiplication operator in the orthonormal Jacobi polynomials is
\begin{equation}
    tp_n = c_n p_{n-1} + a_n p_n + b_n p_{n+1},
\end{equation}
with
\begin{equation}
    a_n := -\frac{B_n}{A_n} \text{, }b_n :=\frac{F^+_n}{A_n}\text{, }c_n := \frac{F^-_n C_n}{A_n}.
\end{equation}

Notice that, by the self-adjointness of multiplication by $t$, we must have $c_n = b_{n-1}$. After some algebraic manipulation, we verify this and also obtain
\begin{equation}
    b_{n,\ell} = 2\sqrt{\frac{(n+1)(n+3/2)(n+\ell+3/2)(n+\ell+2)}{(2n+\ell+2)(2n+\ell+3)^2(2n+\ell+4)}},
\end{equation}
\begin{equation}
    a_{n,\ell} = \frac{\ell(\ell+1)}{(2n+\ell+1)(2n+\ell+3)}.
\end{equation}
Evidently $a_{n,\ell} \geq \widetilde{a}_{n,\ell}$ where
\begin{equation}\label{atilde}
\widetilde{a}_{n,\ell} := \left(\frac{\ell}{2n+\ell +3}\right)^2 = \left(1 - \frac{2n+3}{2n+\ell + 3}\right)^2\ .
\end{equation}
It is clear that $\widetilde{a}_{n,\ell}$ is decreasing in $n$ and increasing in $\ell$. 
It will also be useful to have a bound on $ b_{n,\ell}$ that is simpler to estimate. For this, notice that we can pair the 'close' terms in the numerator, and compare with those in the denominator. We then have

\begin{align}
    b_{n,\ell} \leq & 2\sqrt{\frac{(n+3/2)^2(n+\ell+2)^2}{(2n+\ell+2)^2(2n+\ell+3)^2}} 
    =  \frac{2n+3}{2n+\ell+3} \frac{n+\ell+2}{2n+\ell+2} \nonumber \\
    =  &\bigg(1-\frac{\ell}{2n + \ell + 3}\bigg)\bigg(1 - \frac{n}{2n+\ell+2}\bigg)
    \leq  \bigg(1-\frac{\ell}{2n + \ell + 3}\bigg)\bigg(1 - \frac{n}{2n+\ell+3}\bigg)\ .
\end{align}
We therefore define
\begin{equation}\label{btil}
\widetilde{b}_{n,\ell} :=  \bigg(1-\frac{\ell}{2n + \ell + 3}\bigg)\bigg(1 - \frac{n}{2n+\ell+3}\bigg)\ .
\end{equation}

In order to better understand the behaviour of these sequences, we treat them as functions of continuous variables. Differentiating in $\ell$ we find
\begin{equation}
    \frac{{\partial}}{{\partial}\ell}\widetilde{b}_{n,\ell} = -\frac{(2\ell +6)n + 3\ell + 9}{(2n + \ell +3)^3}\ ,
\end{equation}
so that $\widetilde{b}_{n,\ell} $ is decreasing in $\ell$.  Differentiating in $n$, we see that
\begin{equation}
    \frac{{\partial}}{{\partial}n}\widetilde{b}_{n,\ell} = \frac{2\ell^2 + 3\ell - (6n+9)}{(2n + \ell +3)^3}\ ,
\end{equation}
so that $\widetilde{b}_{n,\ell}$ is increasing in $n$ for 
$$
n \leq  \frac{2\ell^2 + 3\ell -9}{6}\ .
$$

Suppose now $\varphi$ belongs to the $(\ell,m)$ angular momentum sector, so that we can expand it as
\begin{equation}
    \varphi = \sum_{n=0}^{+\infty} \varphi_n f_{n,\ell,m} \ .
\end{equation}

Observing that
\begin{align}
    & (1+2K)\varphi = \sum_{n=1}^{+\infty}(1+2k_{n,\ell})\varphi_n f_{n,\ell,m}\ , \\
    & (1-t)\varphi = \sum_{n=1}^{+\infty}(b_{n-1} \varphi_{n-1} + (1-a_n)\varphi_n + c_{n+1} \varphi_{n+1})f_{n,\ell,m}\ ,
\end{align}
we finally derive
\begin{align}
    \langle (1-t)\varphi,(1+2K)\varphi\rangle = 
    & - \sum_{n=1}^{+\infty}b_{n-1,\ell} (1+2k_{n,\ell})\varphi_{n-1}\varphi_n \nonumber\\
    & + \sum_{n=1}^{+\infty}(1-a_{n,\ell})(1+2k_{n,\ell})\varphi_n^2   \nonumber\\
    & - \sum_{n=1}^{+\infty}b_{n, \ell} (1+2k_{n,\ell})\varphi_{n+1}\varphi_n\ .
\end{align}

Taking the supremum on each term and using Cauchy-Schwarz on the $\varphi_n \varphi_{n\pm 1}$ terms yields, and recalling the upper bound \eqref{kappatildebound} $\kappa_{n,\ell} \leq \widetilde \kappa_{n,\ell}$,
\begin{align}
    \lambda_\ell & \leq \frac14 \bigg[2\sup_n \{ b_{n,\ell} (1+2\kappa_{n,\ell}) \} + \sup_n \{(1-a_{n,\ell})(1+2\kappa_{n,\ell}) \}\bigg]  \nonumber\\
    & \leq \frac14 \bigg[2\sup_n \{ \widetilde{b}_{n,\ell} (1+2 \widetilde \kappa_{n,\ell}) \} + \sup_n \{(1-\widetilde{a}_{n,\ell})(1+2\widetilde \kappa_{n,\ell}) \}\bigg].
\end{align}

We can then turn our attention to the sequences $\widetilde{b}_{n,\ell} (1+2 \widetilde \kappa_{n,\ell})$ and $(1-\widetilde{a}_{n,\ell})(1+2\widetilde \kappa_{n,\ell})$.

\begin{Prop}
For all $\ell \geq 70$, $\lambda_{\ell} \leq 0.73016$. 
\end{Prop}

\begin{proof}
We first compute the derivative in $n$ of $(1-\widetilde{a}_{n,\ell})(1+2\widetilde{ \kappa}_{n,\ell})$,
 which we find
$$
\frac{4\ell^2}{3} \frac{ 4\sqrt{6e} + 3\sqrt{n+\sqrt{\ell}}}{(2n+\ell + 3)^3  \sqrt{n+\sqrt{\ell}}}  - \frac23 \frac{\left(1- \frac{k^2}{(2n+\ell+3)^2}\right)\sqrt{6e}}{(n + \sqrt{\ell})^{3/2}}\ .
$$
which has the same sign as
\begin{equation}\label{abound}
{2\ell^2} \frac{ 4\sqrt{6e} + 3\sqrt{n+\sqrt{\ell}}}{(2n+\ell + 3)^3 }  - \frac{\left(1- \frac{k^2}{(2n+\ell+3)^2}\right)\sqrt{6e}}{n + \sqrt{\ell}}\ .
\end{equation}
Now assume $n\geq \ell^{3/2}$.  Then, we have that ${2\ell^2} \leq (2n+\ell + 3)^{4/3}$, $\sqrt{n+\sqrt{\ell}} \leq (2n+\ell + 3)^{1/2}$,
and
$$
\frac{k^2}{(2n+\ell+3)^2} \leq \frac{1}{4\ell}\ .
$$
Thus, when $n\geq \ell^{3/2}$, the quantity in \eqref{abound} is bounded above by
$$
 \frac{ 4\sqrt{6e/\sqrt{\ell}} + 3}{(n + \sqrt{\ell})^{7/6} }  - \frac{\left(1- \frac{1}{4\ell}\right)\sqrt{6e}}{n + \sqrt{\ell}} \ .
$$
Hence our derivative is negative when $n \geq \ell^{3/2}$ and 
$$
(n + \sqrt{\ell})^{1/6} \geq  \frac{ 4\sqrt{6e/\sqrt{\ell}} + 3}{\left(1- \frac{1}{4\ell}\right)\sqrt{6e}}\ .
$$
For $\ell = 70$, this is satisfied for $n\geq 86$, and $70^{3/2} \leq 586$. Hence, for $\ell =70$, $(1-\widetilde{a}_{n,\ell})(1+2\widetilde{ \kappa}_{n,\ell})$ is decreasing in $n$ for al $n\geq 586$, and these values of $n$ and $\ell$ is it less than $1.22$. 
Checking the finitely many remaining cases, one finds that the maximum occurs at $n=66$, and is less than $1.4351$. 

Proceeding similarly for $2 \widetilde{b}_{n,\ell} (1+2\tilde \kappa_{n,\ell})$ we find the for $\ell =70$, the maximum occurs at $n= 53$, and is no more than $1.4855$. 
Altogether we find,
$$\lambda_{70} \leq \frac{1.4855+1.4351}{4} \leq 0.73016\ .$$
Since it is clear that $b_{n,\ell} (1+2\tilde \kappa_{,\ell})$ and $(1-a_{n,\ell})(1+2\tilde \kappa_{n,\ell})$ are decreasing in $\ell$, this proves that $\lambda_\ell \leq 0.73016$ for all $\ell \geq 70$.
\end{proof}

\subsubsection{Small $\ell$}

By inequality \eqref{kappatildebound}, and ignoring the dependence in $\ell$, we see that for $n\geq151$ we have $\kappa_{n,\ell} \leq 0.23$.
Numerically checking, this holds as well for all $6\leq \ell \leq 50$ and $0\leq n \leq 151$.  Thus, for $5 < \ell < 70$, the norm of $K_\ell$, the restriction of $K$ to the $\ell$ angular momentum sector satisfies
$$
\|K_\ell\| \leq  0.23\ .
$$
It follows from \eqref{crude} that for $5 < \ell < 70$,
$$
\lambda_\ell \leq \frac12(1+2\cdot0.23) = 0.73\ .
$$
Taking into consideration the $\ell$ dependency can also help reduce the amount of cases needed to check.

This leaves only the $\ell = 0,1,2,3,4,5$ cases. Below, we provide a computation for $\ell=0$, making use of the explicit expressions obtained via \eqref{krecurrence} and \eqref{kzero}. The same computation can be carried out analogously for the remaining cases, at the expense of explicitly computing the three term relation factors and the eigenvalues, and taking an appropriate choice of index for a break point.

For $\ell = 0$, the operator $\varphi(v) \mapsto  (1 - \tfrac12 |v|^2)\varphi(v)$  is given in the Jacobi basis by the matrix
$$
X := \frac14 \left[\begin{array} {cccccc} 
\phantom{-}2 & -1 & \phantom{-}0 & \phantom{-}0 & \phantom{-}0 & \dots\\
-1 & \phantom{-}2 & -1 &\phantom{-}0 & \phantom{-}0  & \dots\\
\phantom{-}0 & -1 & \phantom{-}2 & -1 & \phantom{-}0  & \dots\\
\phantom{-}0 & \phantom{-}0 & -1 & \phantom{-} 2 & -1 & \dots \\
\phantom{-}0 & \phantom{-}0 &\phantom{-} 0 &-1& \phantom{-}2  & \dots\\
\phantom{-}\vdots & \phantom{-}\vdots &\phantom{-}\vdots & \phantom{-}\vdots & \vdots &\ddots
\end{array}\right] \ .
$$

Let $Y$ be the diagonal matrix whose $n$th diagonal entry is $1 + 2\kappa_{n+1,0}$.
The upper $5\times 5$ block of $Z := \frac12 (YX +XY)$ is 
$$ \left[\begin{array}{ccccc} 
\phantom{-}\frac12 & -\frac{5}{16} & \phantom{-}0 & \phantom{-}0  & \phantom{-}0 \\
-\frac{5}{16} & \phantom{-}\frac34 & -\frac{21}{80} &\phantom{-} 0 & \phantom{-}0 \\
\phantom{-}0& -\frac{21}{80} & \phantom{-}\frac{3}{10} & -\frac15 & \phantom{-}0 \\
\phantom{-}0 & \phantom{-}0 & -\frac15 & \phantom{-}\frac12  & -\frac27  \\
\phantom{-}0 & \phantom{-}0 & \phantom{-}0 & -\frac27 & \phantom{-}\frac{9}{14}
\end{array}\right]\ .
$$
According to Maple, the largest eigenvalue of  this matrix  is less than $1.0412$.   We also have that 
$$
Z_{6,5} = Z_{5,6} = -\frac{57}{224}\ .
$$
Let $\varphi$ be any unit vector in $\ell^2$, and let $\eta$ be the vector with $\eta_j = \varphi_j$ for $j=1,2,3,4$, and $\eta_j =0$ for $j > 4$. Let $\zeta := \varphi - \eta$. Then
$$\langle \varphi, Z \varphi\rangle =  \langle \eta, \widehat{Z}\eta\rangle
 +    \langle \zeta, Z \zeta\rangle - \frac{57}{224} 2\varphi_4\varphi_5 \ .$$
 We have 
 $$
 \langle \eta, \widehat{Z}\eta\rangle  \leq 1.0412 \|\eta\|^2\ ,
 $$
 and by the Schwarz argument, since the largest relevant eigenvalue of $K$ is $\tfrac{1}{10}$, 
 $$
 \langle \zeta, Z \zeta\rangle  \leq \frac{6}{5}\|\zeta\|^2\ .
 $$
 Thus, $\langle \varphi, Z \varphi\rangle$ is less than the largest eigenvalue of 
 $$
 \left[ \begin{array}{cc}
 1.0412 & -\frac{57}{224}\\  -\frac{57}{224} & \frac65\end{array}\right] \ ,
 $$
 which is no more than $1.388$. Dividing by $2$, we get our bound on the top eigenvalue in the $\ell = 0$ sector, $0.694$.

\section{Conclusion}
Comparing the bounds obtained in each sector, theorem \ref{maintheorem} follows. Obviously, the result is only as good as the worst bound we obtain. While a bound of $0.02$ is certainly achievable as we have seen, it is not clear if this method could yield a much better result.

Nevertheless, we appear to be off by one order of magnitude, which is decent enough, and we saw how the erratic behaviour of the eigenvalues $\kappa_{n,\ell}$ poses an obstacle towards this.

Thus, a future goal would be proving entropy production inequalities with non-degenerate 'constants', first for the conjugate process and possibly for the Kac process itself following Villani's work. This is out of reach at the moment, and careful geometrical analysis is needed due to the state space constraints. 

It is possible, however, that these quantities can be related analogously to their respective spectral gaps, and new insights can be obtained this way.

\renewcommand{\abstractname}{Acknowledgements}
\begin{abstract}
The author would like to thank CMAF-CIO for the opportunity to develop this work and professor José Francisco Rodrigues for bringing this topic to my attention. 

I also would like to extend my deepest gratitude to professor Eric Carlen, my supervisor, for all the patience, clear and intuitive explanations and everything that I have learned regarding functional analysis and mathematical physics for the duration of this project.
\end{abstract}

\clearpage
\appendix

\section{Small $\ell$ computations}\label{smallellcomp}

For $\ell=1,...,5$, we may proceed exactly the same way as for the $\ell=0$ case. Recall that, through the recurrence relation \eqref{krecurrence}, we are able to obtain explicit formulas for eigenvalues in each $\ell$ sector. For instance, we can always write $\kappa_{n,\ell}$ as a linear combination of $\kappa_{n,0}$, $\kappa_{n+1,0}$, ..., $\kappa_{n+\ell,0}$, and we see that
\begin{align}
    \kappa_{n,1} &  = -(\kappa_{n,0} + \kappa_{n+1,0}) \ , \\
    \kappa_{n,2} &  = \frac{n+5/2}{n+2} \kappa_{n,0} + \kappa_{n+1,0} + \frac{n+3/2}{n+2}\kappa_{n+2,0} \ .
\end{align}
Once again, these simplify if we consider separate cases given $n \mod 3$, and analogous formulas hold for $\ell=3,4,5$, but with 4, 5 and 6 terms respectively.

We begin by computing the upper $5\times 5$ block of the $Z$ matrix, as well as its largest eigenvalue, and note the remainder entries that separate it from the lower block (the coefficient of the $\varphi_5 \varphi_6$ term). Finally, we compute the largest relevant eigenvalue $\kappa_{n,\ell}$ for the Schwarz argument.

This procedure is represented in the table below, with fractions denoting the exact values and decimals as close upper bounds.

\begin{center}
 \begin{tabular}{||c c c c c||} 
 \hline
 $\ell$ & Upper block bound & Remainder & Largest $\kappa_{n,\ell}$ & Eigenvalue bound \\ [0.5ex] 
 \hline\hline
 0 & 1.0412 & -57/224 & 1/10  & 0.694 \\ 
 \hline
 1 & 0.946 & -0.2729 & 1/8 & 0.7052 \\ 
 \hline
 2 & 0.895 & -0.2084 & 0.12 & 0.669 \\ 
 \hline
 3 & 0.81 & -0.254 & 0.105 & 0.667 \\ 
 \hline
 4 & 0.784 & -0.2275 & 0.12  & 0.6671 \\ 
 \hline
 5 & 0.754 & -0.206 & 0.1 & 0.6403 \\  [1ex] 
 \hline
\end{tabular}
\end{center}

\section{Python code for the Monte Carlo simulations}
We leave in this appendix the Python code that was used for the Monte Carlo simulations of the $\alpha=2$ case, as the only difference in the Maxwellian case are the constant jump rates $1/N$. Please note that the code is not optimized and vast memory and run time improvements could be made. The helper functions are defined as close as possible to our theoretical conventions and notations.

\begin{lstlisting}
N = 3

def norm_2(v):

    return sum([x**2 for x in v])

def exp_sample(lamb=1):

    u = np.random.rand()
    X = -np.log(1-u)/lamb
    return X

def lamb(v):

    l = [0]*N
    for i in range(N):
        l[i] = (N**2 - (1 + norm_2(v[i]))*N)/( N*(N-1)**2 )
    return l

def sample_jump_time(v):

    params_ = lamb(v)
    jump_time = exp_sample( params_[0] )
    jump_index = 0
    
    for i in range(1,N):
        candidate = exp_sample( params_[i] )
        if candidate < jump_time:
            jump_time = candidate
            jump_index = i
    return jump_time, jump_index  

def beta(v):
    return np.sqrt(3/2 * (1-norm_2(v)))

def conditional_jump_to_uniform3(vel,index):
    
    new_vel = vel
    v = vel[index]/np.sqrt(2)
    
    y1 = np.random.normal(size=3)
    y2 = -y1
    scale = (np.sqrt(norm_2(y1)+norm_2(y2)))/np.sqrt(2)
    
    y1 = y1/scale
    y2 = y2/scale
    
    if index == 0:
        new_vel[1] = beta(v)*y1 - v/np.sqrt(2)
        new_vel[2] = beta(v)*y2 - v/np.sqrt(2)
    elif index == 1:
        new_vel[0] = beta(v)*y1 - v/np.sqrt(2)
        new_vel[2] = beta(v)*y2 - v/np.sqrt(2)
    elif index == 2:
        new_vel[0] = beta(v)*y1 - v/np.sqrt(2)
        new_vel[1] = beta(v)*y2 - v/np.sqrt(2)
        
    return new_vel

n_sim = 1000000
sim_count = 0
frames = [2, 4, 6, 8, 10, 12, 14, 16, 18, 20, 22, 24]

first_sample = []
first_sample2 = []
first_sample3 = []
final_sample = [[] for x in range(12)]
final_sample2 = [[] for x in range(12)]
final_sample3 = [[] for x in range(12)]
index_dist = [[] for x in range(12)]

while sim_count < n_sim:
    
    r = np.random.rand()
    u = np.random.rand()
    
    while u > (1-r):
        r = np.random.rand()
        u = np.random.rand()
        
    v = np.random.normal(size=3)
    v = r*v/np.sqrt(norm_2(v))
    first_sample.append(2*r)
    
    # initialize the remaining coordinates according to the parametrization
    v2 = np.random.normal(size=3)
    v3 = -v2
    scale = (np.sqrt(norm_2(v2)+norm_2(v3)))/np.sqrt(2)
    
    v2 = v2/scale
    v3 = v3/scale
    
    vel = [np.sqrt(2)*v,beta(v)*v2 - v/np.sqrt(2), beta(v)*v3 - v/np.sqrt(2)]
    first_sample2.append(norm_2(vel[1]))
    first_sample3.append(norm_2(vel[2]))
    
    t = 0
        
    for i in range(len(frames)):
        T = frames[i]
        while t<T:
            tau, index = sample_jump_time(vel)
            index_dist[i].append(index)
            t += tau
        
            vel = conditional_jump_to_uniform3(vel, index)
    
        final_sample[i].append(np.sqrt(norm_2(vel[0])))
        final_sample2[i].append(np.sqrt(norm_2(vel[1])))
        final_sample3[i].append(np.sqrt(norm_2(vel[2])))
    
    sim_count += 1

entropy = []

first_dist = [x/2 for x in first_sample]
counts, bins, bars = plt.hist(first_dist, bins=100, density=True)
ent = sum([counts[i]*np.log(counts[i]/(5.09295817894*(bins[i])**2 *
    (1-bins[i]**2)**(1/2)))/100 for i in range(1,len(counts))])
entropy.append(ent)
x = np.linspace(0,1)
plt.plot(x, 5.09295817894*x**2 * (1-x**2)**(1/2))
plt.show()

for i in range(12):
    final_dist = [x/np.sqrt(2) for x in final_sample[i]]
    counts, bins, bars = plt.hist(final_dist, bins=100, density=True)
    ent = [counts[i]/100 * np.log( counts[i] /
    ((bins[i])**2*(1-bins[i]**2)**(1/2)*5.09295817894)) 
    for i in range(1,len(counts))]
    total_ent = sum(ent)
    entropy.append(total_ent)
    plt.plot(x, 5.09295817894*x**2 * (1-x**2)**(1/2))
    plt.show()
    
frames0 = [0, 2, 4, 6, 8, 10, 12, 14, 16, 18, 20, 22, 24]
plt.scatter(frames0, y = np.log(entropy))
\end{lstlisting}

\end{document}